\documentclass[pdflatex,sn-mathphys-num]{sn-jnl}

\usepackage{includes}
\usepackage{macros}

\begin{document}
	
	\title{Constant-time decoding of Gabidulin codes and their generalizations with application to RQC}
	
	\author*[1]{Nicolas Aragon}\email{nicolas.aragon@unilim.fr}
	\author[1]{Chloé Baïsse}
	\author[1]{Anthony Fraga}
	\author[1]{Philippe Gaborit}
	\author[1]{Ilaria Zappatore}
	
	\affil[1]{XLIM, CNRS UMR 7252, Universit\'e de Limoges,
		123, avenue Albert Thomas, 87060 {\sc Limoges Cedex}, {\sc (France)}}
		
		
	\abstract{Gabidulin codes are a rank metric analog of Reed-Solomon codes. They are of great importance in rank-metric cryptography, as they are employed in highly efficient rank-based schemes, such as RQC and the Loidreau cryptosystem. However, until now, there has been no constant-time implementation of these codes, even though such an implementation is crucial for real-life development of cryptosystems. In this paper, we propose the first constant-time decoding algorithm of augmented Gabidulin (AG) codes, a simple variation of Gabidulin codes, in which columns of zeros are added to the end of a Gabidulin code. AG codes are used in the most efficient variants of the RQC cryptosystem. We prove that AG codes decoding can be achieved with quadratic complexity by reducing the decoding of an AG code to the decoding of a shorter standard Gabidulin code. In addition to the description of the complete AG codes decoding procedure, we present a constant-time algorithm for the left division of $q$-polynomials that can be used in every constant-time implementation of Gabidulin code variants. These algorithms are integrated into the RQC-Block-MS-AG scheme, and we evaluate the performance of our implementation through benchmarks. Our results show that our constant-time implementation outperforms the original RQC, though it remains approximately four to five times slower than HQC. However, it achieves ciphertext and key sizes about four times smaller, highlighting an appealing trade-off between performance and compactness.}
		
	\maketitle
	
	\section{Introduction}

\paragraph{Code-based cryptography.} Code-based cryptography \label{key} was first introduced by McEliece \cite{McEliece1978APK} in 1978. For many years, code-based cryptography remained outside mainstream, as number theory-based cryptography enabled highly efficient cryptosystems. Situation shifted with the emergence of the quantum threat posed by Shor's algorithm \cite{shor1994algorithms} in 1994. Following the development of quantum computers in the 2010s, NIST initiated a post-quantum cryptography competition in 2017 to establish quantum-resistant cryptographic standards. After eight years, two encryption schemes were selected: KYBER \cite{fips203}, in 2024, based on lattices, and HQC \cite{hqc-spec}, in 2025, based on coding theory and the Hamming distance. Most prominent code-based systems, including BIKE \cite{aragon2020bike}, HQC \cite{hqc-spec}, and the McEliece cryptosystem \cite{McEliece1978APK}, are based on coding problems involving the Hamming distance.

However, another metric, the rank metric, can also be employed. Introduced in 1985 by Gabidulin \cite{gabidulin1985theory}, the rank metric differs significantly from the Hamming one. In the 2000s, it attracted considerable attention within the coding community due to its applicability to network coding. The rank metric is also relevant for cryptography, as it enables the construction of rank analogues of Reed-Solomon codes, known as Gabidulin codes. These codes were utilized in early cryptosystems, such as the GPT cryptosystem \cite{gpt1991ideals}, which represents an instance of the McEliece cryptosystem using Gabidulin codes. Owing to the difficulty of decoding random codes in the rank-metric, parameter sizes can be much smaller than those required in the Hamming metric. Problems in the rank-metric are generally considered more challenging than those in the Hamming metric, analogous to the increased difficulty of the discrete logarithm problem over elliptic curves compared to $\mathbb{Z}/p\mathbb{Z}$. Initially, the highly structured nature of Gabidulin codes made effective masking difficult, resulting in many rank-metric systems being repeatedly broken and repaired. Consequently, rank-based cryptography, despite its potential, was not regarded as sufficiently secure. This perception began to shift in 2013 with the introduction of LRPC codes and the LRPC cryptosystem \cite{gaborit2013lrpc}, which, unlike previous systems, did not depend on masking Gabidulin codes. Subsequently, in 2016, the RQC cryptosystem was introduced in the same paper as HQC \cite{aguilar2016efficient}. RQC, an analogue of HQC in the rank metric, benefits from the property that no masking is required because the decoding code is public. This development enabled the use of Gabidulin codes in cryptography without concerns regarding the security of masking. In 2017, a new masking technique for Gabidulin codes, involving their combination with LRPC-like codes, was proposed in \cite{loidreau2017new}. This approach appears to be resistant and has been further improved in \cite{aragon2024lowms}.

Rank-metric-based systems such as RQC \cite{rqc-spec} and LRPC \cite{gaborit2013lrpc} were submitted to the NIST competition and advanced to the second round. However, they did not progress to the third round because generic algebraic attacks against rank-metric schemes were improved in \cite{bardet2020algebraic}, compromising the parameters proposed to NIST. Since 2020, algebraic attacks in the rank metric have become better understood and stabilized. Therefore, some improvements have been made to the RQC cryptosystem, including the use of augmented Gabidulin codes (a generalization of Gabidulin codes to which one adds columns of zeros), multiple syndromes \cite{bidoux2023rqc}, and more recently, blockwise errors \cite{song2025blockwise,aragon2024blockwise}. All these improvements have enabled RQC to achieve parameters very close to those obtained for KYBER and significantly lower than those of BIKE or HQC.

\medskip

\paragraph{Constant-time implementation in rank-metric.}

Constant-time implementation is a programming technique that ensures timing variations cannot be correlated with secret data. More specifically, an implementation is constant time if there are no conditional branches and memory access depends on secret data. Despite the name, a constant-time implementation does not ensure the execution time is constant, only that running time does not depend on secret data.

At the beginning of the NIST standardization process, things differed among the diverse post-quantum communities. For instance, since lattice-based cryptography had been well established for years, researchers began to consider constant-time implementations. These were not only a requirement of the NIST competition but, also an essential condition for the secure practical development of cryptographic systems. For the coding community, cryptosystems were less stable. For instance, the HQC team did not propose a constant-time implementation until the third round of the competition. Regarding the rank metric, 'almost' constant-time implementations were proposed for the RQC cryptosystem: for example, the Euclidean division step in \cite{bettaieb2019preventing} is still implemented in variable time.
Since RQC did not reach the third round, all work done during the competition was gathered in the RBC public library \cite{rbclib}. No further progress was made on the constant-time implementation of Gabidulin codes. Consequently, Gabidulin codes are used in the RQC cryptosystem \cite{rqc-spec} and in the recent Loidreau's cryptosystem \cite{loidreau2017new} and their variants, while there is currently no constant-time implementation of these codes. Yet, constant-time implementations for rank-based cryptography are essential.

\paragraph{Our contributions.}
This paper introduces a constant-time decoding algorithm for augmented Gabidulin (AG) codes.

First, we present a reduction of the augmented Gabidulin decoding problem to the decoding problem for standard Gabidulin codes. This reduction enables the use of Loidreau's decoding algorithm \cite{loidreau2005welch} instead of linear algebra, thereby achieving quadratic decoding complexity.

Next, a constant-time implementation of AG code decoding is developed. To achieve this, the core functions of the RBC library for $q$-polynomial arithmetic are modified to make them constant-time. In particular, a constant-time algorithm for the left division of $q$-polynomials, required in Loidreau's decoding algorithm, is introduced. Although our work focuses on AG codes, these modifications to the RBC library may also benefit other variants of Gabidulin codes.

Finally, the proposed decoding algorithm is integrated into the RQC-Block-MS-AG scheme \cite{aragon2024blockwise}, proving promising performance: our implementation is faster than the original RQC and, while approximately four times slower than HQC, achieves ciphertext and key sizes that are about four times smaller.

\paragraph{Organization of this paper.}
Section 2 presents the background on rank-metric, Gabidulin codes and decoding algorithms, and Gabidulin code generalizations. Section 3 describes the reduction of the augmented Gabidulin code decoding problem to the reconstruction problem, and gives a complexity analysis of the resulting decoding algorithm. Section 4 gives details on the implementation of the decoding algorithm and the left Euclidean division in a constant-time context, as well as our modifications and optimizations to the RBC library. Section 5 shows different results of our experiments, including constant-time analysis, comparison of Gabidulin generalizations, and RQC performance.

	\section{Preliminaries}
In this section we recall the main notions used throughout the paper. We first introduce the notations and the basic definitions on rank-metric codes. Then, we introduce $q$-polynomials, Gabidulin codes, and their decoding. Finally, we present augmented Gabidulin codes and we give an overview of the last version of the Rank Quasi-Cyclic (RQC) scheme based on these codes.

\subsection{Notation}
Let $q$ be a prime power, $\Fq$ a finite field of order $q$, and $\Fqm$ the extension field of $\Fq$ of degree $m$.
\noindent
In this paper, we represent vectors by lowercase bold letters $\xv, \yv$ and matrices by uppercase bold letters $\Gv, \Hv$. Polynomials and $q$-polynomials (see Section~\ref{subsec: q-poly_and_Gabs}) are denoted by capital letters $A(X), B(X)$, and sometimes we omit the indeterminate for simplicity.

\noindent

\subsection{Background on rank-metric codes }
Given $\xv=(x_1, \ldots, x_n) \in \Fqm^n$, we define its \textit{support} as,
\[
	\Supp(\xv) \eqdef \langle x_1, \ldots, x_n\rangle_{\Fq},
\]
where $\langle x_1, \ldots, x_n\rangle_{\Fq}$ denotes the $\Fq$-linear span of the $x_i$'s, viewed as vectors in~$\Fq^m$.

\noindent
The \textit{rank weight} of $\xv$ is defined as,
\[
	\wrank{\xv}\eqdef \dim(\Supp(\xv))
\]
and the \textit{rank distance} of two vectors $\xv, \yv \in \Fqm^n$ is,
\[
	d(\xv,\yv)\eqdef \wrank{\xv-\yv}.
\]
\begin{remark}
	If we fix an $\Fq$-basis $\mathcal{B}=\{\beta_1, \ldots, \beta_m\}$ of $\Fqm$, any vector $\xv=(x_1, \ldots, x_n)\in \Fqm^n$ can be written as a matrix $\Mv_\mathcal{B}(\xv)\in \Fq^{m \times n}$ by representing any element $x_i\in \Fqm$ as a column vector whose entries are its coefficients in the basis $\mathcal{B}$. Under this point of view, the rank of $\xv$ is the rank of $\Mv_\mathcal{B}(\xv)$.
\end{remark}
We can now introduce the following,

\begin{definition}[$\Fqm$-linear code]
	An $\Fqm$-linear code $\C$ of length $n$ and dimension $k$ is an $\Fqm$-vector subspace of $\Fqm^n$.
	It is referred to as an $[n, k]_{q^m}$ code.
\end{definition}
An $[n, k]_{q^m}$ code $\C$ can be represented by:
\begin{itemize}
	\item  a full-rank \textit{generator matrix} $\Gv \in \Fqm^{k \times n}$, \ie a matrix whose rows form a basis of $\C$,
	\item a full-rank \textit{parity check matrix} $ \Hv \in \Fqm^{(n-k)\times n}$, \ie a generator matrix of the \textit{dual} code $\C^{\perp}$ of $\C$, which is the orthogonal of $\C$ with respect to the standard inner product in $\Fqm$.
\end{itemize}

The \textit{minimum distance} of $\C$ is defined as,
\[
	d_{min}(\C) \eqdef \min_{\xv\in \C\setminus \{0\}} \{\wrank{\xv}\}.
\]

We now recall the rank-metric decoding problem.
\begin{definition}[\bf{Decoding}($\yv$, $\C$, $t$)]\label{def:decoding} Let $\yv \in \Fqm^n$, $\C$ an $\Fqm$-linear code, and $t$ a positive integer.
	Find, if it exists, a codeword $\cv \in \C$ and an error vector $\ev\in \Fqm^n$ such that,
	\[
		\yv=\cv+\ev \textit{ and } \wrank {\ev}\leq t.
	\]
\end{definition}

\subsection{The ring of \texorpdfstring{$q$}{q}-polynomials and Gabidulin codes}\label{subsec: q-poly_and_Gabs}
In rank-metric cryptography, a central notion is that of $q$-polynomials, a special class of polynomials introduced by Ore in \cite{ore1933special}. They are defined as $\Fqm$-linear combinations of the monomials $X, X^q, X^{q^2}, \ldots, X^{q^i}$, respectively denoted by $X, X^{[1]}, X^{[2]}, \ldots, X^{[i]}$. In fact, these monomials are the successive iterations of the Frobenius automorphism

\[\begin{array}{cccc}
		\theta : & \Fqm & \to     & \Fqm  \\
		         & x    & \mapsto & x^{q}
	\end{array}\]

\noindent
A $q$-polynomial over $\Fqm$ is a polynomial of the form,
\[ P(X) = \sum_{i = 0}^{d_P} p_iX^{[i]} \]
with $p_i\in \Fqm$. If $P$ is nonzero and $d_P\neq 0$, the integer $d_P$ is called the \textit{q-degree} of $P$, and it is denoted $\deg_q(P)$.

\noindent
We equip the space of $q$-polynomials with a \textit{non-commutative ring structure}, denoted $\FqmXq$, where addition is defined term-wise as in classical polynomial rings, and the multiplication is defined by composition and extended by $\Fqm$-linearity as follows:
\[
	\forall i, j \in \mathbb{N},\ \forall a \in \Fqm, \qquad X^{[i]} X^{[j]} = X^{[i+j]}, \qquad X^{[i]} a = a^{q^i} X^{[i]}.
\]

\noindent
The ring $\FqmXq$ is a \textit{left} and \textit{right} \textit{Euclidean domain}.

\noindent

\subsubsection{The left euclidean division for \texorpdfstring{$q$}{q}-polynomials} \label{sec:div}
In this paper, we will focus on the left division, which will be used later in the decoding of Gabidulin codes.

\noindent
Given $A, B \in \FqmXq$ with $B \ne 0$, $d_A \eqdef \deg_q(A)$ and $d_B\eqdef \deg_q(B)$ and $d_A \geq d_B$. We aim to find unique $Q, R \in \FqmXq$ such that:
\[
	A = B Q + R \quad \text{with } \deg_q(R) < d_B.
\]
Left euclidean division is performed using an adaptation of the classical euclidean algorithm to the ring of $q$-polynomials. We describe this algorithm below, as we present a constant-time version of it later.

Let $B(X) = \sum_{i=0}^{d_B} b_iX^{[i]}$, where $d_B=\deg_q(B)$. The algorithm is initialized with quotient $Q_0=0$ and remainder $R_0=A$.
At each iteration $i$, provided that $\deg_q(R_i) \geq d_B$, the Euclidean algorithm cancels the leading term of the current remainder $R_i$ by computing
\begin{equation}
	\label{eq:rem_update}
	R_{i+1}(X) = R_i(X) -  B(X) \circ cX^{[\deg_q(R_i)-d_B]}
\end{equation}
The coefficient $c$ is chosen so that the leading coefficient of $B(X) \circ cX^{[\deg_q(R_i)-d_B]}$ is equal to the leading coefficient of $R_i$, \ie $lc(R_i)$.
Indeed,
\begin{equation}\label{eq:comp_B}
	b_{d_B}X^{[d_B]}\circ cX^{[\deg_q(R_i)-d_B]} = b_{d_B}c^{[d_B]} X^{[\deg_q(R_i)]}
\end{equation}
$c$ satisfies \[b_{d_B}c^{[d_B]} = lc(R_i)\] \\
Therefore, $c$ must satisfy $b_{d_B}c^{[d_B]}=lc(R_i)$ and so,
\begin{equation}\label{eq:coeff_c}
	c =( lc(R_i) b_{d_B}^{-1})^{[-d_B]} 
\end{equation}
Here, for $x\in \Fqm$, the notation $x^{[-t]}$ denotes the image of $x$ under the inverse of the $t$-th Frobenius automorphism.
The monomial $cX^{[\deg_q(R_i)-d_B]} $ is also added to the current quotient, \ie
\[
	Q_{i+1}(X)= Q_{i}(X)+cX^{[\deg_q(R_i)-d_B]}.
\]

The algorithm terminates when the current remainder $R_i$ satisfies $\deg_q(R_i) < d_B$.
The total number of iterations depends on the gap between the $q$-degrees of $A$ and $B$, as well as on the possible degree drops of $R_i$ between two successive iterations.

\begin{proposition}\label{prop: complexity_div}
	The left euclidean division of $A$ by $B$ requires $\Ocomp(d_A^2+d_Am)$ operations over $\Fqm$.
\end{proposition}
\begin{proof}
	At each iteration of the algorithm, the leading term of the current remainder is cancelled. Therefore, the algorithm performs at most $d_A-d_B+1$ iterations. During one iteration, the coefficient $c$ is given by \eqref{eq:coeff_c}. 
	The inverse of $b_{d_B}$ can be precomputed once before the beginning of the division. Since the Frobenius automorphism has order $m$ over $\Fqm$, inverse Frobenius is taken modulo $m$. Thus, $x^{[-d_B]} = x^{[s]}$ where $0\leq s\leq m-1$ satisfies $s = -d_B \bmod m$. In the implementation considered in this work, this value is computed using successive Frobenius applications and therefore requires at most $m-1$ such operations.
	The composition \eqref{eq:comp_B} involves $d_B + 1$ coefficients. Updating the remainder and the quotient consequently requires $\Ocomp(d_B)$ additions and multiplications over $\Fqm$. Hence, each iteration costs $\Ocomp(d_B + m)$ operations over $\Fqm$.
	Multiplying this cost by the maximum number of iterations gives $\Ocomp((d_A - d_B + 1)(d_B + m))$, and the thesis follows from $d_A \geq d_B$
\end{proof}
\subsubsection{Gabidulin codes}
We can now define Gabidulin codes, a family of $\Fqm$-linear codes introduced in \cite{gabidulin1985theory}. These are the rank metric analogues of Reed-Solomon codes.

\begin{definition}[Gabidulin codes]
	Let $k,n$ be two positive integers, with $k<n\leq m$ and let $\gv = (g_1, \ldots, g_n)\in \Fqm^n$ such that $\wrank{\gv}=n$, the \textit{Gabidulin code} of length $n$ and dimension $k$ is defined by
	\[
		\Gabcode{\gv}{n,k,m} \eqdef \left\{ F(\gv)=(F(g_1), \ldots, F(g_n))  \mid F(X)\in \FqmXq, \hspace{0.2cm} \deg_q(F)<k \right\},
	\]
	The vector $\gv$ is called the evaluation vector.
\end{definition}
Gabidulin codes are \textit{maximum rank distance} (MRD) codes, \ie they achieve the rank metric Singleton bound. Their minimum distance is $d_{min}(\Gabcode{\gv}{n,k,m})=n-k+1$. Consequently, they can be uniquely decoded from any error of rank at most $\tau_0 \eqdef \left\lfloor\frac{n-k}{2}\right\rfloor$.

\subsection{On the decoding of Gabidulin codes}  \label{subsec:decodingGab}
Several decoding algorithms have been proposed for Gabidulin codes. As for Reed-Solomon codes in the Hamming metric, these algorithms can be divided into two approaches. The syndrome-based approach determines a $q$-polynomial whose root space is the error support, typically by solving a \emph{syndrome-based key equation} using an analogue of the extended euclidean algorithm, or of the Berlekamp-Massey algorithm \cite{gabidulin1985theory,gabidulin1991fast,roth1991maximum,richter2004error,richter2004fast}. This strategy is suitable for high-dimensional Gabidulin codes.

The second approach works directly with the received word and recovers the $q$-polynomial corresponding to the transmitted message by solving another \emph{key equation} or equivalently, a \emph{reconstruction problem}, as in Welch-Berlekamp or Gao-type decoding for Reed Solomon codes, \cite{loidreau2005welch, wachterzeh2012fastdecoding, augot2018generalized}.
In this work, we focus on the second approach.

\begin{definition}[\bf{Reconstruction}($\yv$, $\gv$, $n$, $k$, $t$)] \label{def: LR} Let $\yv=(y_1, \ldots, y_n) \in \Fqm^n$, $\gv=(g_1, \ldots, g_n) \in \Fqm^n$. Let $k$ be a positive integer and $t\geq 0$.
	Find a nonzero pair of $q$-polynomials $(V, N)$ such that,
	\[
		\begin{cases}
			V(y_i)=N(g_i), \hspace{0.2cm} 1\leq i \leq n \\
			\deg_q(V)\leq t                              \\
			\deg_q(N)\leq k+t-1
		\end{cases}
	\]

\end{definition}
In \cite{loidreau2005welch}, Loidreau showed the following fundamental result.
\begin{theorem}\label{thm:recoDecGab}
	Let $\yv=F(\gv)+\ev$, where $F\in\FqmXq$ satisfies $\deg_q(F)<k$ and $\wrank{\ev}\leq t$. Assume that $t \leq \tau_0$. Then for any nonzero solution $(V,N)$ of {\bf{Reconstruction}($\yv$, $\gv$, $n$, $k$, $t$)}, $F$ is the quotient of the left euclidean division of $N$ by $V$ and $(F(\gv), \yv-F(\gv))$ is the unique solution of \bf{Decoding}($\yv$, $\Gabcode{\gv}{n,k,m}$, $t$).
\end{theorem}
When $r\leq \tau_0$, the Reconstruction problem admits a \emph{unique minimal solution} $(V,N)$, up to multiplication by a nonzero scalar. Here, minimality means that $V$ has the smallest possible $q$-degree among all nonzero solutions. More precisely, $V$ is a scalar multiple of the annihilator $q$-polynomial of the error support, and $N=V \circ F$, where $F$ is the $q$-polynomial corresponding to the transmitted codeword. Thus, after this minimal solution has been computed, $F$ is recovered py performing the left euclidean division of $N$ by $V$.
\begin{algorithm}
	\caption{Decoding algorithm for Gabidulin codes}
	\label{alg:decode_gab}
	\begin{algorithmic}[1]
		\Require{$\Gabcode{\gv}{n,k,m}$, $\yv \in \Fqm^n$, the received word with error of rank $t\leq \lfloor \frac{n-k}{2} \rfloor$}
		\Ensure{$(\cv, \ev)$ a solution of \textbf{Decoding}($\yv$, $\Gabcode{\gv}{n,k,m}$, $t$)}
		\State $(V, N) \gets$ \Call{SolveReconstruction}{$\yv$, $\gv$, $n$, $k$, $\lfloor \frac{n-k}{2} \rfloor$}
		\State $F \gets$ \Call{LeftEuclideanDivision}{$N$,$V$}
		\State \Return $(F(g_1), \ldots,F(g_n)), \yv - (F(g_1), \ldots,F(g_n) )$
	\end{algorithmic}
\end{algorithm}

A natural approach to solve the \textbf{Reconstruction} problem is to set up a linear system where the unknowns are the coefficients of the $q$-polynomials $(V, N)$. This can be solved using linear algebra in time $O(n^{\omega})$, where $\omega$ is the exponent of matrix multiplication.
Loidreau \cite{loidreau2005welch} proposed a more efficient algorithm for solving the \textbf{Reconstruction} problem, later generalized and improved by Augot \textit{et al.} \cite{augot2018generalized}, leading to a decoding algorithm with quadratic complexity in the code length.

Loidreau's algorithm iteratively constructs two pairs of $q$-polynomials, denoted by $(N_0, V_0)$ and $(N_1, V_1)$. After the $i$-th iteration, for $k \leq i\leq n$, both pairs satisfy the following interpolation conditions:
\begin{equation}
	\label{eq:loid}
	\forall j \leq i, \left\{
	\begin{array}{l}
		N_0 (g_j) - V_0(y_j) = 0 \\
		N_1(g_j)- V_1(y_j) = 0
	\end{array}
	\right.
\end{equation}
The algorithm starts by defining the following initial pairs:
\begin{equation}
	\label{eq:init_loidreau}
	\begin{array}{ll}
		(N_0(X), V_0(X)) & = (\mathcal{A}(X), 0) \\
		(N_1(X), V_1(X)) & = (\mathcal{I}(X), X)
	\end{array}
\end{equation}
thus guaranteeing that \eqref{eq:loid} is satisfied for the first $k$ positions.
Note that,
\begin{itemize}
	\item $\mathcal{A} \eqdef \mathcal{A}_{\langle g_1, \ldots, g_k\rangle}(X)$ is the \textit{annihilator} $q$-polynomial of the $\Fq$- subspace $\langle g_1, \ldots, g_k\rangle$, \ie $\mathcal{A}(\vv) = 0$ for any $\vv$ in this subspace;
	\item $\mathcal{I}(X)$ is the unique \textit{interpolation} $q$-polynomial satisfying $\mathcal{I}(g_i) = y_i$ for $1 \leq i \leq k$.
\end{itemize}
For each step $k+1 \leq i \leq n$, both pairs of $q$-polynomials $(V_0, N_0)$ and $(V_1, N_1)$ are updated according to their discrepancies at the point $(g_i, y_i)$. The update rules are chosen so that, after the $i$-th iteration, both new pairs satisfy the interpolation conditions at the first $i$ positions.

At the end of the procedure, the algorithm selects one of the two pairs satisfying the degree constraints of the \textbf{Reconstruction} problem. We refer to \cite{loidreau2005welch,augot2018generalized} for the correctness proof of the algorithm.

Bettaieb et al. showed in \cite{bettaieb2019preventing} that there exists a correlation between the weight of the error to be corrected and the number of interpolation steps in Loidreau's algorithm.
They used this correlation to build two theoretical timing attacks against RQC and proposed a constant-time version of the algorithm as a countermeasure.

\subsection{Augmented Gabidulin Codes}
Gabidulin codes have given rise to several generalizations tailored to meet specific cryptographic requirements. In this paper, we focus on one of them, the \textit{augmented} Gabidulin (AG) codes \cite{bidoux2023rqc}, which, as we will see later, play a central role in the RQC cryptosystem.

\begin{definition}[Augmented Gabidulin codes]\label{def:AGcode} Let $k \leq n' \leq m < n$ be positive integers, $\gv = (g_1, \ldots, g_{n'}) \in \Fqm^{n'}$ such that $\wrank{\gv}=n'$, and $\overline{\gv} \in \Fqm^n$ be the vector $\gv$ completed by $n-n'$ zeros on the right.

	\noindent
	The AG code $\AugGabcode{\gv}{n,n',k,m}$ is the $[n, k]_{q^m}$ code defined as \[\AugGabcode{\gv}{n,n',k,m} = \left\{F(\overline{\gv}) \mid F(X) \in \FqmXq, \deg_q(F)<k \right\}.\]
\end{definition}

The minimum rank distance of AG codes is $n'-k+1$, since appending zero coordinates does not change the rank weight of their codewords. The additional zero coordinates provide \emph{partial information} on the error support and can therefore increase the correctable error rank. In particular, if they provide an $\varepsilon$-dimensional subspace of the error support, called \emph{erasure support}, unique decoding is possible for errors of rank at most $\left\lfloor \frac{n' - k + \varepsilon}{2} \right\rfloor$ (see\cite[Proposition 1]{bidoux2023rqc}). 

\subsection{The RQC cryptosystem using augmented Gabidulin codes}

In this section, we recall the RQC-Block-MS-AG variant of the RQC cryptosystem, \cite{aragon2024blockwise}.
Compared with the original RQC scheme, this variant incorporates three main improvements.
First, following \cite{bidoux2023rqc}, it replaces Gabidulin codes with AG codes, resulting in a smaller parameter $m$.
Second, again following the proposal in \cite{bidoux2023rqc}, it employs the multiple syndromes (MS) strategy \cite{aguilar2022lrpc}. Finally, it uses errors with a blockwise rank-support structure, a principle introduced in \cite{song2025blockwise}.

\begin{definition}[Blockwise $\ell$-error]
	Let $\nv = (n_1, \ldots, n_\ell)$ and $\rv = (r_1, \ldots, r_\ell)$ be tuples of positive integers, and let $n = \sum_{i = 1}^{\ell} n_i$. A blockwise $\ell$-error with parameters $\nv$ and $\rv$ is a vector $\ev=(e_1, \ldots, e_\ell) \in \Fqm^n$  such that for every $\ 1 \leq i \leq \ell$, one has $e_i \in \Fqm^{n_i}$ with $\wrank{e_i} = r_i$ and $\Supp(e_i) \cap \Supp(e_j) = \{0\}$ for all distinct $i,j$, with $1\leq i,j \leq \ell$ .
\end{definition}

To describe the cryptosystem in Figure \ref{fig:rqc}, we use the following notations:
\begin{itemize}
	\item $\mathcal{S}^{\nv}_{\rv}(\Fqm)$ denotes the set of blockwise errors with parameters $\nv$ and $\rv$.
	\item $\mathcal{S}^{N\times n}_{\rv}(\Fqm)$, the set of block matrices $\Mv = (\Mv_1 | \ldots |\Mv_\ell)\in \Fqm^{N \times n}$ where each block $\Mv_i\in \Fqm^{N \times n_i}$ has support of dimension $r_i$ over $\Fq$, and the supports of distinct blocks are pairwise disjoint.
	\item We identify vectors $\vv \in \Fqm^{n_2}$ with polynomials of degree less than $n_2$ in $\Fqm[X]/(P)$. For $\vv \in \Fqm^{n_2}$ and $\Mv = (\mv_1 | \ldots |\mv_{n_1})\in \Fqm^{n_2 \times n_1}$, where $\mv_i$ denotes the $i$-th column of $\Mv$, we define $\vv \cdot \Mv$ as the matrix whose $i$-th column is the coefficient vector of the product of the polynomials associated with $\vv$ and $\mv_i$ reduced modulo $P$. Here, $P$ is an irreducible polynomial of degree $n_2$ over $\Fqm$.
\end{itemize}
We also define the map
\[ \begin{array}{cccc}
		\Fold: & ( \Fqm^{n_2})^{n_1}        & \to     & \Fqm^{n_2 \times n_1}          \\
		       & (\vv_1, \ldots, \vv_{n_1}) & \mapsto & (\vv_1^T| \ldots| \vv_{n_1}^T).
	\end{array}\]
Its inverse is denoted by $\Unfold$.

\begin{figure}
	\fbox{
		\begin{minipage}{\textwidth}
			KeyGen($1^\lambda$)
			\begin{enumerate}
				\item $\gv \xleftarrow{\$} \mathcal{S}^{m}_{m}(\Fqm)$, $\hv \xleftarrow{\$} \Fqm^{n_2}$, $(\xv, \yv) \xleftarrow{\$} \mathcal{S}^{(n_2, n_2)}_{(r_{\xv}, r_{\yv})} (\Fqm)$
				\item $\sv \gets \xv + \hv \cdot \yv \mod P$
				\item \textbf{return }$\pk = (\gv, \hv, \sv)$, $\sk = (\xv, \yv)$
			\end{enumerate}
			
			\bigskip
			Encrypt($\pk, \mv$)
			\begin{enumerate}
				\item Compute a generator matrix $\Gv \in \Fqm^{k \times n_1n_2}$ of $\AugGabcode{\gv}{n_1n_2, m, k, m}$, where $\overline{\gv} = (\gv, 0_{n_1n_2-m})$
				\item $(\Rv_1, \Rv_2, \Ev) \xleftarrow{\$} \mathcal{S}^{n_2 \times (n_1, n_1, n_1)}_{(r_1,r_2, r_{\ev})} (\Fqm)$
				      \item$\Uv \gets \Rv_1 + \hv \cdot \Rv_2$
				      \item$\Vv \gets \Fold(\mv\Gv) + \sv \cdot \Rv_2 + \Ev$
				\item \textbf{return }$\Cv = (\Uv, \Vv)$
			\end{enumerate}
			
			\bigskip
			Decrypt($\sk, \Cv)$
			\begin{enumerate}
				\item \textbf{return }$\mathsf{Decode}(\Unfold(\Vv - \yv \cdot \Uv))$
			\end{enumerate}
		\end{minipage}}
	\caption{The RQC-Block-MS-AG cryptosystem}
	\label{fig:rqc}
\end{figure}

Parameters for 128-bit security of the RQC-Block-MS-AG cryptosystem \cite{aragon2024blockwise} are given in Table \ref{tab:rqc_params}.

\begin{table}[ht]
	\centering
	\caption{Parameters for RQC-Block-MS-AG-128}
	\label{tab:rqc_params}
	\begin{tabular}{|c|c|c|c|c|c|c|c|c|c|c|c|c|}
		\hline
		$q$ & $m$ & $n_2$ & $k$ & $\varepsilon$ & $r_{\xv}$ & $r_{\yv}$ & $r_1$ & $r_2$ & $r_{\ev}$ & $n_1$ & DFR  & $\pk + \ct$ (in kB) \\
		\hline
		2   & 43  & 52    & 3   & 32            & 4         & 4         & 4     & 4     & 4         & 2     & -145 & 1.43                \\
		\hline
	\end{tabular}
\end{table}
	
	\section{Decoding augmented Gabidulin codes}\label{sec:decodingAG}
In this section, we focus on the decoding of augmented Gabidulin codes. We provide an explicit reduction of this decoding problem to the Reconstruction problem introduced in Definition~\ref{def: LR}. This reduction, allows us to apply Loidreau's algorithm~\cite{loidreau2005welch} and solve the decoding problem efficiently.

\subsection{A reconstruction-based decoding of AG codes}

The decoding algorithm exploits the zero positions of an AG code. Since the last $n-n'$ coordinates of every codeword are zero, the corresponding coordinates of the received vector consist of error components and therefore reveal a subspace of the error support.
By annihilating this known subspace, we reduce the decoding problem for an AG code to the decoding of a standard Gabidulin code with shorter length, larger dimension, and lower error rank.

\noindent
Let $\AugGabcode{\gv}{n, n', k, m}$ be an AG code as in Definition~\ref{def:AGcode} where $\gv\eqdef(g_1, \ldots, g_{n'})\in \Fqm^{n'}$ is the evaluation vector and $\overline{\gv}=(g_1, \ldots, g_{n'}, 0,\ldots, 0)\in \Fqm^n$.
Write the received vector as
\[
	\yv = (\yv_1, \yv_2) \in \Fqm^n,
\]
where $\yv_1 = (y_1, \ldots, y_{n'}) \in \Fqm^{n'}$ and $\yv_2 = (y_{n'+1}, \ldots, y_n)\in \Fqm^{n-n'}$. Assume that $\yv = \cv +\ev$, where $\cv = (F(\gv), 0, \ldots, 0)$ for some $F \in \FqmXq$ of $\deg_q(F)<k$, and where $\ev =(\ev_1, \ev_2)$ has rank $r\leq t$. Since the last $n-n'$ nonzero components of every codeword are zero, $\yv_2 = \ev_2$. Hence $\Supp(\yv_2)=\Supp(\ev_2)$ is a known subspace of $E \eqdef \Supp(\ev)$. Let $E_2$ be an $\varepsilon$-dimensional subspace of $\Supp(\yv_2)$ called \emph{support erasure}. Then $0\leq \varepsilon \leq \min\{\dim (\Supp(\yv_2)),r\}$. We assume that $\varepsilon \leq n'-k$ and

\begin{equation}\label{eq: dec_radius_AG}
	t \leq \floor{\frac{n'-k+\varepsilon}{2}}.
\end{equation}

Let $V \in\FqmXq$ be the monic annihilator $q$-polynomial of $E$. Since $V$ vanishes on every component of the error, we have
\[
	V(\yv)=V(F(\overline{\gv})).
\]
where the evaluations are applied componentwise.
Hence, setting $N\eqdef V\circ F$, the pair $(V,N)$ satisfies
\textbf{Reconstruction}$(\yv, \overline{\gv}, n, k, t)$, \ie
\[
	\begin{cases}
		V(y_i)=N(g_i), \hspace{0.2cm} 1\leq i \leq n \\
		\deg_q(V)\leq t                              \\
		\deg_q(N)\leq k+t-1
	\end{cases}
\]

Now, let $V_2$ be the monic annihilator polynomial of the support erasure $E_2$. Since $E_2\subseteq E$, there exists a unique monic polynomial $W\in \FqmXq$ such that $V=W\circ V_2$. Moreover, $\deg_q(W)=r-\varepsilon\leq t-\varepsilon$. Define $\yv_1'\eqdef V_2(\yv_1)$. Since $V = W \circ V_2$, the interpolation equations on the first $n'$ coordinates can be rewritten as
\begin{equation}\label{eq:recAG}
	\begin{cases}
		W(y'_{1,i})=N(g_i), \hspace{0.2cm} 1\leq i \leq n' \\
		\deg_q(W)\leq t-\varepsilon                        \\
		\deg_q(N)\leq k+t-1
	\end{cases}
\end{equation}
Since $k+t-1 = (k+\varepsilon)+(t-\varepsilon)-1$, this is an instance of \textbf{Reconstruction}$(\yv_1',\gv, n', k+\varepsilon, t-\varepsilon)$.
This can be interpreted as the decoding of a standard Gabidulin code. Indeed, by the $\Fq$-linearity we get $\yv_1' = (V_2\circ F)(\gv)+V_2(\ev_1)$.
Moreover, $\deg_q(V_2\circ F)<k+\varepsilon$ and
$\wrank{V_2(\ev_1)}\leq r-\varepsilon$.
Hence, $\yv_1'$ is a received word for $\Gabcode{\gv}{n',k+\varepsilon,m}$
corrupted by an error of rank at most $r-\varepsilon$. Finally, \eqref{eq: dec_radius_AG} is equivalent to
\[
	t-\varepsilon
	\leq
	\left\lfloor
	\frac{n'-(k+\varepsilon)}{2}
	\right\rfloor,
\]
so the reduced problem lies in the unique-decoding radius of
$\Gabcode{\gv}{n',k+\varepsilon,m}$.

Thus we can conclude the following.
\begin{theorem}\label{thm:decAGcodes} Let $\yv=(\yv_1,\yv_2)$ be a received word for
	$\AugGabcode{\gv}{n,n',k,m}$, corrupted by an error of rank $r\leq t$. Let $E_2 \subseteq \Supp (\yv_2)$ be a support erasure of dimension $\varepsilon$, and let $V_2$ be its monic annihilator $q$-polynomial. Assume that $\varepsilon \leq n'-k$ and $t\leq
		\lfloor\frac{n'-k+\varepsilon}{2}\rfloor$.
	Then the decoding problem
	\textbf{Decoding}$(\yv,\, \G^+_{\gv}(n, n', k, m),\, t)$
	reduces to a standard Gabidulin decoding problem
	\textbf{Decoding}$(V_2(\yv_1), \G_{\gv}(n', k + \varepsilon, m),\, t - \varepsilon)$.
\end{theorem}

We can now sum up the entire decoding (see Algorithm~\ref{alg:decoding-augmented-gabidulin}).
\paragraph*{Step 1. Computing the support erasure and its annihilator.}
Since the last $n-n'$ coordinates of every AG codeword are zero, $\yv_2$ coincides with the corresponding error components. In the cryptographic setting, the maximum error rank $t$ and the required erasure dimension $\varepsilon$ are fixed by the parameters of the scheme so that \eqref{eq: dec_radius_AG} is satisfied.
Let $\rho\eqdef\dim (\Supp(\yv_2))$. If $\rho\geq\varepsilon$, the decoder selects an $\varepsilon$-dimensional subspace $E_2\subseteq\Supp(\yv_2)$ and computes its monic annihilator $q$-polynomial $V_2$, for instance using the iterative procedure described in \cite{augot2018generalized}.
If $\rho<\varepsilon$, the revealed subspace is too small to guarantee the correction of an error of rank $\leq t$ and decryption fails. The \emph{decryption failure rate} is the probability that $\dim(\Supp(\yv_2))<\varepsilon$. After fixing an $\Fq$-basis of $E=\Supp(\ev)$, the last $n-n'$ error components can be represented by an $r \times (n-n')$ matrix over $\Fq$, where $r=\dim (E)$. Under the uniform random-matrix model, the decryption failure rate is the probability that this matrix has rank smaller than $\varepsilon$, \ie
\[
	1-
	\frac{1}{q^{r(n-n')}}
	\sum_{i=\varepsilon}^{\min\{r,n-n'\}}
	\prod_{j=0}^{i-1}
	\frac{(q^r-q^j)(q^{n-n'}-q^j)}
	{q^i-q^j}.
\]
\paragraph*{Step 2. Decoding a shorter Gabidulin code.} Given $V_2$, we can compute $\yv_1' = V_2(\yv_1)$, and by Theorem~\ref{thm:decAGcodes}, we can just apply Loidreau's algorithm to solve the corresponding \textbf{Reconstruction}$(\yv'_1,\gv, n', k+\varepsilon, t-\varepsilon)$ problem as in (\ref{eq:recAG}) and find $(W,N)$.
\paragraph*{Step 3. Recovering the codeword and its associate $q$-polynomial.} Since $N=W\circ V_2\circ F$, where $F$ is the $q$-polynomial corresponding to the codeword, we recover $F$ through two successive left euclidean divisions.

\begin{algorithm}
	\caption{Decoding algorithm of the augmented Gabidulin codes}
	\label{alg:decoding-augmented-gabidulin}
	\begin{algorithmic}[1]
		\Require{$\AugGabcode{\gv}{n,n',k,m}$, $\yv\in\Fqm^n$ the received word with error of rank at most $t\leq \lfloor \frac{n'-k+\varepsilon}{2} \rfloor$, $\varepsilon$ the dimension of the support erasure}
		\Ensure{$(\cv,\ev)$ solving \textbf{Decoding}$(\yv,\, \AugGabcode{\gv}{n,n',k,m},\, t)$, or $\perp$ if decoding fails}

		\State $\mathcal{B}_2\gets \Call{Echelonize}{\yv_2}$ \Comment{$\mathcal{B}_2=\{\mathcal{B}_{2,1},\ldots,\mathcal{B}_{2,\rho}\}$ is a basis of $\Supp(\yv_2)$}
		\State $\rho\gets |\mathcal{B}_2|$
		\If{$\rho<\varepsilon$}
		\State \Return $\perp$ \Comment{The revealed support has dimension smaller than $\varepsilon$}
		\EndIf
		\State $V_2(X)\gets \mathcal{A}_{\langle \mathcal{B}_{2,1},\ldots,\mathcal{B}_{2,\varepsilon} \rangle}(X)$
		\For{$i=1$ to $n'$}
		\State $y'_i\gets V_2(y_i)$
		\EndFor
		\State $W,N\gets \Call{SolveReconstruction}{\yv', \gv, n',k+\varepsilon, t-\varepsilon}$
		\State $N_2\gets \Call{LeftEuclideanDivision}{N,W}$
		\State $F\gets \Call{LeftEuclideanDivision}{N_2,V_2}$
		\State $\cv\gets (F(\gv), 0, \ldots, 0)$
		\State \textbf{return} $\cv$, $\yv-\cv$
	\end{algorithmic}
\end{algorithm}
\paragraph*{Difference with the AG decoding algorithm of \cite{bidoux2023rqc}.}
In \cite{bidoux2023rqc}, the decoding of augmented Gabidulin codes is performed by deriving interpolation equations equivalent to~\eqref{eq:recAG} and solving the resulting linear system over $\Fqm$ in the coefficients of the unknown $q$-polynomials.

In this work, we identify these equations as an instance of the standard reconstruction problem for a Gabidulin code of length $n'$, dimension $k+\varepsilon$, and decoding radius $t-\varepsilon$. This interpretation allows us to exploit the algebraic structure of the reconstruction problem and  apply Loidreau's  algorithm instead of solving the system by  linear algebra. As shown in the next section, the decoding algorithm has quadratic complexity.

\subsection{Complexity analysis}\label{sec:complexity}

\begin{theorem}\label{th:complexity-ag-decoding}
 When \[(n'-n)^2 = \max((n-n')^2, \varepsilon(n'+\varepsilon),n'^2, (k+t)^2+(k+t)m,(k+\varepsilon)^2+(k+\varepsilon)m)\]
 Algorithm~\ref{alg:decoding-augmented-gabidulin} requires $\Ocomp((n-n')^2)$ operations over $\Fqm$.
\end{theorem}

\begin{proof}
	We analyze the complexity of the three main steps of Algorithm~\ref{alg:decoding-augmented-gabidulin}.

	The first part of the algorithm computes the support erasure. With respect to a fixed $\Fq$-basis of $\Fqm$, the vector $\yv_2 \in \Fqm^{n-n'}$ is represented by an $m \times (n-n')$ matrix over $\Fq$. Computing a basis of $\Supp{(\yv_2)}$ consists in performing Gaussian elimination on this matrix. Assuming $n-n'\geq m$, this requires $\Ocomp(m(n-n')^2)$ operations over $\Fq$. Here, we express all costs in extension-field-equivalent operations over $\Fqm$. Thus, this step requires $\Ocomp((n-n')^2)$
	Once an $\varepsilon$-dimensional space of $\Supp{(\yv_2)}$ has been selected, its monic annihilator $q$-polynomial $V_2$ can be computed in $\Ocomp{(\varepsilon^2)}$ operations over $\Fqm$ using the iterative procedure of \cite{augot2018generalized}. Evaluating $V_2$ at the $n'$ coordinates of $\yv_1$, requires $\Ocomp{(n'\varepsilon)}$ operations over $\Fqm$. It follows that the computation and evaluation of $V_2$ require in total $\Ocomp(\varepsilon(n'+\varepsilon))$ operations.

	The second part consists in solving the reconstruction problem with parameters $(\yv_1', \gv, n', k+\varepsilon, t-\varepsilon)$. Since $k+\varepsilon \leq n'$, this step requires $\Ocomp(n'^2)$ operations over $\Fqm$.

	In the last part, the algorithm performs two euclidean left divisions.
	By Proposition~\ref{prop: complexity_div}, the steps 10 and 11 require $\Ocomp((k+t)^2+(k+t)m)$ and $\Ocomp((k+\varepsilon)^2+(k+\varepsilon)m)$ operations over $\Fqm$, respectively.
	
	The costs for each step have now been determined, and we obtain Theorem \ref{th:complexity-ag-decoding}.
\end{proof}

\begin{remark}
	The AG code considered in RQC-MS-Block-AG is $\AugGabcode{\gv}{n_1n_2, m, k, m}$, as stated in Figure \ref{fig:rqc}. Therefore, in the notations of Algorithm \ref{alg:decoding-augmented-gabidulin}, we have $n = n_1n_2$ and $n'=m$. Using the parameters given in Table \ref{tab:rqc_params}, we obtain \[(n-n')^2 \geq (k+t)^2+(k+t)m \geq (k+\varepsilon)^2+(k+\varepsilon)m \geq \varepsilon(n'+\varepsilon) \geq n'^2\]
	where $t = \floor{\frac{n'-k+\varepsilon}{2}}$. 
	
	Hence, in the case of RQC-MS-Block-AG, the condition of Theorem \ref{th:complexity-ag-decoding} is verified and the complexity of Algorithm \ref{alg:decoding-augmented-gabidulin} is $\Ocomp((n-n')^2)$. 
\end{remark}

\subsection{Comparison with another generalization of Gabidulin Codes} \label{sec:decoding-comparison}

Another generalization of Gabidulin codes, called Extended Gabidulin (EG) codes, was introduced by Berger and Ourivski in 2009 \cite{berger2009construction}. Their original motivation was to construct new MDS codes derived from classical Gabidulin codes.

\begin{definition}[Extended Gabidulin codes] Let $k, n', n, m$ be positive integers such that $k \leq n' \leq \min\{n,m\}$, and let $\gv = (g_1, \ldots, g_n)\in \Fqm^n$ with $\wrank{\gv} = n'$.

	\noindent
	The Extended Gabidulin code $\ExtGabcode{\gv}{n,n',k,m}$ is the $[n, k]_{q^m}$-code defined as
	\[\ExtGabcode{\gv}{n,n',k,m} = \left\{F(\gv)=\big(F(g_1), \ldots, F(g_n)\big) \mid F(X)\in \FqmXq, \hspace{0.2cm} \deg_q(F)<k\right\}.\]
\end{definition}

We observe that when $n'=n \leq m$, the coordinates of $\gv$ are linearly independent over $\Fq$, and the resulting EG code is a Gabidulin code. In particular, EG codes allow the evaluation vector to have rank strictly smaller than its length, and hence to contain $\Fq$-linearly dependent coordinates.

AG codes form a particular subclass of EG codes. In particular, if the evaluation vector of an EG code is obtained by appending $n-n'$ zero coordinates to an $\Fq$-linearly independent vector $\gv$, then the considered EG codes is an AG code. This property is especially relevant in cryptographic settings such as RQC, where controlling the support of errors is crucial (see Section~\ref{subsec:decodingGab}).

In \cite{song2025interleaved}, Song et al.\ propose using EG codes in RQC and show that their decoding can be formulated as an instance of the \textbf{Reconstruction}$(\yv, \gv, n, k, t)$ problem. They then solve this problem by applying Loidreau's algorithm, obtaining a decoding complexity $\Ocomp(n^2)$.

In contrast, after recovering and annihilating a support erasure, we showed that decoding AG codes reduces to a reconstruction problem with smaller parameters. The resulting decoding algorithm has complexity $\Ocomp((n-n')^2)$, as established in Section~\ref{sec:complexity}. Decoding is therefore faster for AG codes than for EG codes under the parameters listed in Table~\ref{tab:rqc_params} (see Section~\ref{sec:results} for further details).
	
	\section{Implementation details}
We implement the \texttt{RQC-Block-MS-AG} scheme (Figure~\ref{fig:rqc}) using the RBC library, employing Algorithm~\ref{alg:decoding-augmented-gabidulin} for decoding.
We target a constant-time implementation.

For the \Call{SolveReconstruction}{} procedure (see line 9 of Algorithm \ref{alg:decoding-augmented-gabidulin}), we rely on the decoding algorithm for Gabidulin codes already available in the RBC library.
It is an implementation of Loidreau's algorithm largely inspired by the constant-time version described in \cite{bidoux2023rqc}.
This implementation, although the number of iterations is fixed, is not fully constant-time, since some of the operations on $q$-polynomials are still executed in variable time.
So, to achieve a complete constant-time implementation of RQC-Block-MS-AG, we need to implement all the required $q$-polynomial operations in constant-time. This constitutes a part of the contributions of this paper.

To provide some background on the RBC library in the context of rank metric operations, the field extension considered is $\Fm$, defined as $\F[X]/(P)$, where $P$ is a sparse irreducible polynomial of degree $m$.
The elements and the arithmetic operations defined on this field form the core layer of the RBC library.
An element $x$ of $\Fm$ is represented as a vector $(x_0, \ldots, x_{m-1})$ in $\F^m$, in a data structure called \texttt{rbc\_elt}.
Arithmetic operations on \texttt{rbc\_elt} are carried out through polynomial arithmetic modulo $P$.
A vector of elements in $\Fm$ is represented by the \texttt{rbc\_vec} structure, consisting of a pointer to a fixed-size array of \texttt{rbc\_elt}.
Finally, the \texttt{rbc\_qpoly} structure, which represents $q$-polynomials, is composed of three components: $max\_degree$, the upper bound of the $q$-degree; $degree$, the current $q$-degree; and $values$, an \texttt{rbc\_vec} of size $max\_degree + 1$ containing the coefficients of the polynomial.

\subsection{Basic arithmetic operations}
For many operations on $q$-polynomials implemented in the RBC library, the number of iterations depends on the $q$-degree of the input $q$-polynomials.
Consequently, most of our modifications aim to make this number of iterations depend on the upper bound of the $q$-degree rather than on the $q$-degree itself.
Instead of relying on the $max\_degree$ component of a $q$-polynomial as this upper bound, we choose to make it an explicit input parameter of each function.
This approach allows the user to select a custom upper bound that may be smaller than the actual $max\_degree$ of the $q$-polynomial.

This is particularly useful in the decoding algorithm for AG codes. Indeed, in Loidreau's algorithm, it has been shown that in the interpolation steps $i \in \{k+1, \cdots, n\}$, the $q$-degrees of the polynomials $V_0, N_0, V_1$, and $N_1$ are upper bounded by values that depend on the step index $i$ (see Proposition 11 in \cite{loidreau2005welch}).
For AG codes, the same bounds hold if we replace the parameter $k$ by $k+\varepsilon$, as explained in Section~\ref{sec:decodingAG}.
This modification to the handling of $q$-degrees was the only change required to make the following $q$-polynomial operations constant-time:

\begin{itemize}
	\item addition of $q$-polynomials,
	\item left multiplication of a $q$-polynomial by a scalar,
	\item product of two $q$-polynomials,
	\item evaluation of a $q$-polynomial, and
	\item the $q$-exponentiation of a $q$-polynomial.
\end{itemize}

At the end of each of these functions, the $q$-degree of the resulting $q$-polynomial must be updated.
To ensure this step is performed in constant-time, we implement a dedicated constant-time function for computing the $q$-degree of a $q$-polynomial.
In this function, the coefficients of the input $q$-polynomial are scanned sequentially, starting from the coefficient corresponding to the highest possible $q$-degree ($max\_degree$).
A variable storing the resulting degree is updated at every iteration using \textit{masks} (Algorithm~\ref{alg:deg-qpoly-ct}): when a nonzero coefficient is encountered, the mask preserves the current value of the variable, preventing further changes.
Since the loop always iterates over all the coefficients, regardless of the actual $q$-degree, the execution time remains constant.

We use the same technique in the function that retrieves the leading coefficient of a $q$-polynomial.

\begin{algorithm}[ht]
	\caption{Constant-time $q$-degree computation}
	\label{alg:deg-qpoly-ct}
	\begin{algorithmic}[1]
		\Require{$A(X)\in\FqmXq$ of $q$-degree upper bounded by $max$}
		\Ensure{$\deg_q(A)$}
		\State $d\gets max$
		\State $mask\gets 0$
		\For{$i=max$ to $0$}
		\State $mask\gets mask\,\vee\,(A_i\neq0)$
		\State $d\gets (1-mask)\cdot (i-1)+mask\cdot d$
		\EndFor
		\State \textbf{return} $d$
	\end{algorithmic}
\end{algorithm}

\subsection{The annihilator and interpolator \texorpdfstring{$q$}{q}-polynomials}

The decoding of AG codes involves computing both annihilator and interpolator $q$-polynomials.
In the second step of Algorithm~\ref{alg:decoding-augmented-gabidulin}, the annihilator $q$-polynomial $V_2$ -- which vanishes on $E_2$, the subspace spanned by the support erasures -- is computed.
Moreover, at the beginning of Loidreau's algorithm, the $q$-polynomials $N_0$ and $N_1$ are initialized respectively as an annihilator $q$-polynomial and an interpolator $q$-polynomial (see Equation~\eqref{eq:init_loidreau}).

To compute $V_2$, we develop a constant-time version of the RBC library function that computes the annihilator $q$-polynomial of a given subspace.
For the initialization of $N_0$ and $N_1$, we keep the choice of the RBC library, which uses a function to compute and return both the annihilator and the interpolator $q$-polynomials. We adapt this function to a constant-time version as well.
Both constant-time variants of these functions follow exactly the same computational steps as their original algorithms, but employ the constant-time arithmetic operations on $q$-polynomials. \\

\subsection{Constant-time left division on \texorpdfstring{$q$}{q}-polynomials}
To decode AG codes, $q$-polynomial euclidean left divisions are required (see lines 11 and 12 of Algorithm~\ref{alg:decoding-augmented-gabidulin}).

Originally, these divisions were performed by using the Euclidean algorithm for the left division of $q$-polynomials described in Section \ref{sec:div}.
We propose a constant-time version of this algorithm.

Let us consider the left division of a $q$-polynomial $A$ by $B(X) = \sum_{i=0}^{d_B} b_iX^{[i]}$, where $d_B=\deg_q(B)$. As a reminder, the goal is to find the unique pair $(Q,R)$ of $q$-polynomials satisfying $A = BQ+R$ with $\deg_q(R) < d_B$.

In the standard Euclidean algorithm, the number of iterations is upper bounded by $\deg_q(A)+1$.
This bound is reached in the worst case, when $\deg_q(B)=0$ and the $q$-degree of the remainder $R$ decreases by exactly one at each iteration.
Thus, our constant-time algorithm has to execute at least $\deg_q(A)+1$ iterations.

Since the $q$-degree of $A$ must remain secret, we design our algorithm to always perform $max_A+1$ iterations, where $max_A$ denotes the value of the $max\_degree$ parameter chosen when initializing $A$.
The core idea is to update $Q$ and $R$ as in the classical Euclidean algorithm explained above, while performing \textit{dummy operations} whenever no real update should occur.

Let $i$ be a loop variable, decreasing by $1$ from $max_A$ down to $0$.
In the classical Euclidean algorithm, the monomial used to update $Q$ and $R$ is $cX^{[\deg_q(R)-d_B]}$. In our constant-time version, it is replaced with $cX^{[i]}$ and dummy operations are performed whenever $i \neq \deg_q(R)-d_B$.
This situation occurs in two cases:
\begin{itemize}
	\item \textit{Before the first valid update.} In the classical algorithm, the first meaningful update occurs when $i=\deg_q(A) - d_B$. Since our constant-time loop starts from $i=max_A$, all iterations such that $ max_A \geq i > \deg_q(A)- d_B$ do not correspond to real updates and must therefore be replaced by dummy operations. Valid updates start only when $i\leq \deg_q(A)-d_B$.
	\item \textit{When the remainder's degree drops by more than one.} During the execution of the algorithm, the $q$-degree of the remainder $R$ may decrease by more than one between two consecutive iterations. When such “degree jumps” occur, some iterations of the loop no longer correspond to a valid update, because the remainder has already lost the term that would normally be canceled at that position. In these cases, the algorithm must perform dummy operations so that the number of iterations remains constant. Formally, this happens whenever $i>\deg_q(R)-d_B$.
\end{itemize}

Moreover, when $d_R < d_B$ at the beginning of a loop iteration, the target $q$-polynomials $Q$ and $R$ have already been completely determined, and therefore must not be modified.
In this case,  $\deg_q(R)  - d_B < 0$ and since $i \geq 0$, this also implies that $d_R - d_B - i < 0$.
Note that this situation is encountered when we check for degree jumps of $R$.

Finally, legitimate operations on $Q$ and $R$ occur only when both $\deg_q(A)  - d_B - i  \geq 0 $ and $\deg_q(R)  - d_B - i \geq 0$ are simultaneously satisfied.
This can be easily handled by using a flag variable that determines whether the current iteration should execute a valid or a dummy operation.

\medspace
In order to maintain a constant number of steps, we have implemented several algorithmic optimizations, described in the following sections. Each of them targets a specific computational bottleneck of the constant-time left division algorithm for $q$-polynomials.

\paragraph*{Combined computation of $q$-degree and leading coefficient.}
The first optimization consists of computing both the $q$-degree and the leading coefficient of $B$ and $R$ within a single constant-time function inspired by Algorithm~\ref{alg:deg-qpoly-ct}. It is straightforward to observe that these two values can be retrieved in a single loop.
For the $q$-polynomial $R$, which is recomputed at each iteration, this improvement saves up to $max_A \cdot max_R$ iterations, where $max_R$ denotes the upper bound of the $q$-degree of $R$.

\paragraph*{Inverse of the iterated Frobenius.} A second optimization concerns the computation of the inverse of the Frobenius iterated a number of times equal to the $q$-degree of a $q$-polynomial (see Algorithm \ref{alg:ifrob-ct}).
This operation is required at each step $i$ to compute $c = (lc(R_i)b_{d_B}^{-1})^{[-d_B]}$.
Let $e = (lc(R_i)b_{d_B}^{-1})$ and $max_B$ denote the upper bound of the $q$-degree of $B$.
The following equality holds:
\begin{equation}
	e^{[-d_B]} = (e^{[-max_B]})^{[max_B-d_B]}
\end{equation}
Our implementation exploits this property by first computing $e^{[-max_B]}$, followed by $(max_B-d_B)$ Frobenius iterations and $d_B$ dummy iterations to keep the number of steps unrelated to $d_B$, resulting in a total of $(-max_B\mod m)+max_B=m$ Frobenius computations.

\begin{algorithm}[ht]
	\caption{Constant-time inverse of the iterated Frobenius}
	\label{alg:ifrob-ct}
	\begin{algorithmic}[1]
		\Require{$e \in \Fqm$, $A(X) \in \FqmXq$ of $q$-degree upper bounded by $max$}
		\Ensure{$e^{[-\deg_q(A)]}$}
		\State $f \gets e^{[-max]}$
		\State $tmp \gets f$
		\State $mask = 0$
		\For{$i=max$ to $0$}
		\State $tmp \gets tmp^q$
		\State $mask \gets mask \vee (a_i \neq 0)$
		\State $f \gets (1-mask)\cdot f + mask \cdot tmp$
		\EndFor
		\State\textbf{return} $f$
	\end{algorithmic}
\end{algorithm}

\paragraph*{Product with a monomial.}The third optimization is the implementation of a dedicated function to multiply a $q$-polynomial with a monomial (see Algorithm \ref{alg:mul-qpoly-mon-ct}).
This operation, used to compute $B(X) \circ cX^{[i]}$ at each step $i$, would otherwise require scanning all coefficients of $B(X)$.
However, since a monomial has only one nonzero coefficient, the product can be simplified as:
\begin{equation}
	B(X) \circ cX^{[i]}= \sum_{j=0}^{max_B}  b_j\cdot c^{[j]}X^{[i+j]}
\end{equation}
This formulation avoids unnecessary iterations and makes it possible to perform the multiplication through a single pass over the coefficients of $B(X)$.

We define a monomial by its coefficient and corresponding $q$-degree, which serve as inputs of our function.
\begin{algorithm}[ht]
	\caption{Constant-time product of $q$-polynomial by a monomial}
	\label{alg:mul-qpoly-mon-ct}
	\begin{algorithmic}[1]
		\Require{$A(X)  \in \FqmXq$ of $q$-degree upper bounded by $max$, $c \in \Fqm$, $i \in \mathbb{N}$}
		\Ensure{$A(X) \cdot cX^{[i]}$}
		\State $P \gets 0$
		\State $p_i \gets a_0 \cdot c$
		\For{$j=1$ to $max$}
		\State $c \gets c^q$
		\State $p_{i+j \mod m} \gets p_{i+j \mod m} + a_j \cdot c $
		\EndFor
		\State\textbf{return} $P$
	\end{algorithmic}
\end{algorithm}

\paragraph*{Update of the quotient.}The last optimization concerns the update of the quotient $Q$.
At step $i$, instead of performing the operation $Q \gets Q + flag \times cX^{[i]}$ on the full $q$-polynomials, we directly update the $i$-th coefficient of $Q$ as follows: $q_i \gets q_i + flag \cdot c$.
This optimization avoids unnecessary assignments and thus eliminates up to $max_Q$ coefficient updates per step, where $max_Q$ is the upper bound on the $q$-degree of $Q$.
Furthermore, if we denote $max_M$ as the $q$-degree upper bound of $cX^{[i]}$, representing it only by its coefficient and $q$-degree -- rather than as a $q$-polynomial -- we save $max_M$ elements of $\Fqm$ in memory, as well as $max_M$ initializations and $max_M$ de-allocations.

Finally, combining all the proposed optimizations leads to our constant-time left-division algorithm of $q$-polynomials presented in Algorithm~\ref{alg:left-div-ct}.

\begin{algorithm}[ht]
	\caption{Constant-time left division algorithm}
	\label{alg:left-div-ct}
	\begin{algorithmic}[1]
		\Require{$A(X),B(X)\in\FqmXq$, with $\deg_q(A)\geq\deg_q(B)$, and degrees upper bounded by $max$}
		\Ensure{$Q(X),R(X)$ such that $A(X)=B(X)Q(X)+R(X)$ and $\deg_q(R)<\deg_q(B)$}
		\State $R\gets A$
		\State $Q\gets 0$
		\State $d_A\gets \mathrm{deg\_ct}(A)$ \Comment{From \ref{alg:deg-qpoly-ct}}
		\State $d_B, b_{d_B} \gets \mathrm{deg\_lc\_ct}(B)$   \Comment{Same idea as \ref{alg:deg-qpoly-ct}}
		\State $ib_{d_b}\gets b_{d_B}^{-1}$
		\State $s \gets d_A - d_B - max$ \Comment{prevent leaks on degree gaps}
		\For{$i=max$ to 0}
		\State $d_R, b_{d_R} \gets \mathrm{deg\_lc\_ct}(R)$   \Comment{Same idea as \ref{alg:deg-qpoly-ct}}
		\State $jump\gets d_R-d_B-i$\Comment{prevent leaks on degree falls}
		\State $flag\gets 1- \big((s<0)\vee(jump<0)\big)$
		\State $e \gets ib_{d_B}\cdot r_{d_R}$
		\State $c\gets \mathrm{inverse\_frobenius\_ct}(e, d_B)$ \Comment{Compute $e^{[-d_B]}$ with \ref{alg:ifrob-ct}}
		\State $q_i\gets q_i+ flag\cdot c$
		\State $P \gets \mathrm{qpoly\_monomial\_mul\_ct}(B, c, i)$ \Comment{Compute $B(X) \times c X^{[i]}$ with \ref{alg:mul-qpoly-mon-ct}}
		\State $R\gets R+flag \cdot P$
		\State $s \gets s+1$
		\EndFor
		\State \textbf{return} $(Q,R)$
	\end{algorithmic}
\end{algorithm}\label{sect:implementation}
	
	\section{Results} \label{sec:results}

We perform experiments over our constant-time implementation of the \texttt{RQC-Block-MS-AG} scheme, the decoding algorithm of augmented Gabidulin codes, and left euclidean division. The results support the claim that the implementation is indeed constant-time. They also show the performances of the KEM implementations, which is compared to the previous implementation of RQC, and other KEM implementation. We also provide a comparison between decoding time performance of AG codes and EG codes.

The benchmarks have been performed on a machine equipped with an Intel\textsuperscript{\circledR} Core\texttrademark\ i7-13700 CPU for which the Turbo Boost feature has been disabled during the tests. Programs are compiled using \texttt{gcc} and use the \texttt{openssl} library as a provider for hash functions. Files required to reproduce these benchmarks are publicly available on \url{https://github.com/nicolas-aragon/rqc-constant-time}.

\subsection{Comparison between AG codes and EG codes}

EG codes were introduced in 2009 \cite{berger2009construction} but no decoding algorithm was proposed for these codes at that time. Later in 2023,  AG codes were introduced in \cite{bidoux2023rqc} and a decoding algorithm was proposed. AG codes and EG codes are equivalent codes, so in terms of distance property there is no difference between these two families.

However, in practice, having additional zero columns in AG codes rather than additional random combinations of columns of the Gabidulin code for EG codes is clearly an advantage in terms of decoding, as shown in Section \ref{sec:decodingAG}. Thus, in practice, the specific subfamily of AG codes withing the EG codes family seems to be the most relevant to use in terms of efficiency for the decryption of the RQC cryptosystem and its variations.

Indeed, both theoretical and practical complexities shown in Table \ref{tab:aug-vs-ext} confirm this claim: decoding EG codes requires three times more multiplications than an equivalent AG code, and two times more CPU cycles.

\begin{table}[h]
	\centering
	\caption{Theoretical and practical complexity of decoding generalized Gabidulin codes with parameters $[104,3]$ over $\mathbb{F}_{2^{43}}$ and $\wrank{\gv}=43$}
	\label{tab:aug-vs-ext}
	\begin{tabular}{|c|c|c|c|c|}
		\hline
		\textbf{Code} & \textbf{Additions} & \textbf{Multiplications} & \textbf{Frobenius} & \textbf{CPU cycles} \\ \hline
		AG code       & 15500              & 12200                    & 8800               & 675K                \\ \hline
		EG code       & 23500              & 34300                    & 11200              & 1.50M               \\ \hline
	\end{tabular}
\end{table}

\subsection{Constant-time}

A constant-time implementation requires that no conditional branch or memory access depends on secret data. However, it was shown in \cite{bettaieb2019preventing} that a correlation exists between the error weight and the decoding time of Loidreau's algorithm. The same dependency might thus appear if our AG decoding algorithm is implemented improperly, which is why we analyzed our implementation using the \texttt{Valgrind} instrumentation framework to check for potential leaks.

\paragraph*{Constant-time analysis.} We used Valgrind's \texttt{memcheck} tool to detect data-dependent branching and memory accesses in our program. The analysis relies on marking variables containing sensitive data as uninitialized before running the program, a technique commonly referred to as \emph{data poisoning}. As with the spread of a poison, any variable computed from poisoned data becomes poisoned in turn. Whenever a poisoned variable is used in a conditional branch or as a memory access, \texttt{memcheck} raises a warning. In our case, the received word is secret, as are the $q$-degrees of the $q$-polynomials $V$ and $N$, which correlate with the error weight. We therefore poisoned these variables, along with all other degrees and flags that may be influenced by them.

\paragraph*{Observations.} We observe that poisoning the coordinates of the received word is sufficient for the taint to propagate to all coordinates of the computed $q$-polynomials. Moreover, during the interpolation step, the flags tracking the $q$-degrees of these polynomials become poisoned accordingly, following the same propagation pattern.\\
As discussed in Section~\ref{sect:implementation}, several $q$-polynomial functions of the RBC library use the degree as a conditional branch parameter, in particular to bound the \texttt{for} loops. Guided by this constant-time analysis, we modified the function from the library so that they no longer uses any secret-dependent degree as a loop bound. After these modifications, running the same Valgrind analysis over the full range of error weights raises no further warnings, confirming that the implementation is free of the timing leakage identified in \cite{bettaieb2019preventing}.

\subsection{Performances}

We compare our work to other KEMs: the previous RQC implementation from the NIST proposal \cite{rqc-spec}, the latest HQC implementation \cite{hqc-spec}, and the ML-KEM implementation available in the \texttt{liboqs} library. Our proposal is labeled as RQC-Block-MS-AG in the results table.

\begin{table}[ht]
	\centering
	\caption{Average CPU cycles for performing the Keygen, Encaps, and Decaps operations for different KEMs (128 bits of security).}
	\label{tab:results_cpu}
	\begin{tabular}{|c|c|c|c|c|}
		\hline
		\textbf{Algorithm (128 bits)} & \textbf{Keygen} & \textbf{Encaps} & \textbf{Decaps} & \textbf{Size pk + ct (kB)} \\
		\hline
		RQC-Block-MS-AG               & 300K            & 580K            & 1.65M           & 1.43                       \\
		\hline
		RQC (NIST proposal)           & 595K            & 685K            & 2.75M           & 5.49                       \\
		\hline
		HQC                           & 77K             & 145K            & 325K            & 6.67                       \\
		\hline
		ML-KEM                        & 26K             & 28.5K           & 34K             & 1.57                       \\
		\hline
	\end{tabular}
\end{table}

Table \ref{tab:results_cpu} presents the performance of RQC-Block-MS-AG compared to standardized KEMs and the previous version of RQC. In addition to providing a fully constant-time implementation for the RQC cryptosystem, our implementation outperforms the previous implementation from the NIST standardization process for every operation. When compared to HQC, we observe a performance trade-off: while our implementation is approximately four to five times slower, it compensates with significantly more compact parameters. Specifically, RQC-Block-MS-AG achieves public key and ciphertext sizes roughly four times smaller than HQC. Further optimizations of the implementation could reduce the execution time gap with HQC while maintaining constant-time properties. 
	
	\section{Conclusion}
In this paper, we present the first constant-time implementation of a decoding algorithm for AG codes. This decoding algorithm achieves a better complexity compared to \cite{aragon2024blockwise}, exploiting Loidreau's reconstruction. This theoretical improvement allows us to show that AG codes outperform EG codes for the same parameters, making them more relevant for decryption in the RQC cryptosystem and its variations.

We developed a constant-time implementation of the decoding algorithms for both Gabidulin and AG codes, as well as the left euclidean division operation. Through extensive timing analysis and validation with \texttt{Valgrind}, we demonstrated that our implementation satisfies constant-time requirements, eliminating vulnerabilities to timing attacks. Our RQC-Block-MS-AG implementation significantly outperforms the previous RQC implementation. While it is approximately four to five times slower than the optimized HQC implementation, this performance difference is compensated by a fourfold reduction in public key and ciphertext sizes, offering an attractive trade-off. Furthermore, additional optimizations of the implementation could reduce the execution time gap while maintaining constant-time properties.

The techniques developed for AG code decoding can be directly applied to optimize Gabidulin code decoding, since AG decoding generalizes the Gabidulin case. Future work could explore syndrome decoding approaches for high-dimensional Gabidulin codes, which may offer substantial performance improvements while preserving the compact parameter sizes that make rank-metric based cryptography particularly promising for resource-constrained environments.

\section*{Declarations}

\subsection*{Funding}

This work was supported by the PIA4 X7PQC project and the PEPR PQ-TLS project. 
	
	\bibliography{biblio}
	
\end{document}